\documentclass[11pt]{article}%

\usepackage{amssymb,amsfonts,textcomp,amsmath,amsthm}
\usepackage{latexsym}
\usepackage{lscape}
\usepackage[nospace, compress]{cite}
\usepackage[noend,ruled]{algorithm2e}
\usepackage{enumitem}

\usepackage{graphicx} 
\usepackage{latexsym}
\usepackage{multirow}
\usepackage{rotating}
\usepackage{color,soul}
\usepackage{comment}

\newcommand{\fixlist}{\addtolength{\itemsep}{-5pt}}
\newcommand{\fourvote}[4]{#1 \succ #2 \succ #3 \succ #4}
\newcommand{\fivevote}[5]{#1 \succ #2 \succ #3 \succ #4 \succ #5}

\newcommand{\cc}{{{{\mathrm{CC}}}}}

\newcommand{\sntv}{{{{\mathrm{SNTV}}}}}
\newcommand{\bloc}{{{{\mathrm{Bloc}}}}}
\newcommand{\kborda}{{{{k\hbox{-}\mathrm{Borda}}}}}
\newcommand{\rep}{{{\mathrm{rep}}}}

\newcommand{\pos}{{{{\mathrm{pos}}}}}
\newcommand{\borda}{{{{\mathrm{B}}}}}

\def\row#1#2{{#1}_1,\ldots ,{#1}_{#2}}

\title{\bf Properties of Multiwinner Voting Rules}

\author{Edith Elkind\\
        University of Oxford\\
        United Kingdom
 \and  Piotr Faliszewski\\
       AGH University\\
       Krakow, Poland
\and Piotr Skowron\\
     University of Warsaw\\
     Poland
\and Arkadii Slinko\\
     University of Auckland\\
     New Zealand
}

\newtheorem{theorem}{Theorem}
\newtheorem{definition}{Definition}
\newtheorem{proposition}[theorem]{Proposition}

\newtheorem{corollary}[theorem]{Corollary}
\newtheorem{example}{Example}

\newcommand{\np}{{\mathrm{NP}}}

\newcommand{\naturals}{{{\mathbb{N}}}}

\newcommand{\calR}{{{\mathcal{R}}}}

\begin{document}

\maketitle

\begin{abstract}
  The goal of this paper is to propose and study properties of
  multiwinner voting rules which can be consider as generalisations of single-winner scoring voting rules.  We consider SNTV,
  Bloc, $k$-Borda, STV, and several variants of Chamberlin--Courant's
  and Monroe's rules and their approximations. We identify two broad natural classes of
  multiwinner score-based rules, and show that many of the existing rules can be captured
  by one or both of these approaches.  We then formulate a number of
  desirable properties of multiwinner rules, and evaluate the rules we
  consider with respect to these properties.
\end{abstract}

\section{Introduction}\label{sec:introduction}
\noindent
There are many situations where societies need to select a small set
of entities from a larger group. For example, in indirect democracies
people choose representatives to govern on their behalf, companies
select groups of products to promote to their
customers~\cite{bou-lu:c:chamberlin-courant,bou-lu:c:compression}, web
search engines decide which pages to display for a given
query~\cite{dwo-kum-nao-siv:c:rank-aggregation}, and applicants for a
job (e.g., a tenure-track position at a university) are short-listed
prior to conducting interviews.  For all these tasks we need formal
rules to perform the selection, and the desirable properties of such
rules may depend on the task at hand.  We view these selection rules
as multiwinner voting rules which, given individual preferences,
output groups of winners (which we call
\emph{committees}).  %

Multi-winner elections are even more ubiquitous than single-winner
ones, but much less studied. They were implicitly considered under the
umbrella of choice functions
\cite{arr:b:polsci:social-choice,fis:b:social-choice} but in this
model the size of the elected committee could not be controlled and it
was Debord~\cite{deb:j:prudent} and Felsenthal and
Maoz~\cite{fel-mao:j:norms} who introduced several $k$-choice
functions that elect committees of size exactly $k$, and investigated
their properties. However even within this narrower framework several
quite distinct models happily coexist, which makes a simultaneous
study of them difficult.  One obvious classification is based on the
type of the input.  There are preference-based rules (for which inputs
are sequences of linear orders; see, e.g., the work of Brams and
Fishburn~\cite{bra-fis:b:polsci:voting-procedures}), approval-based
rules (for which inputs are sequences of dichotomies; see, e.g., the
overview of Kilgour~\cite{kil-handbook}), tournament-based rules (for
which inputs are either tournaments or weighted tournaments; see the
book of Laslier for a general overview of tournament-based
rules~\cite{las:b:tournament-solutions}).  With a few exceptions,
majority of papers that study (properties of) multiwinner rules are
devoted to approval-based
rules~\cite{kil-handbook,kil-mar:j:minimax-approval,aziz-elk-etc,bra-kil-san:j:minimax-approval,car-kal-mar:c:minimax-approval}
and to rules based on various forms of the Condrocet
principle~\cite{fis:j:condorcet-committee,geh:j:condorcet-committee,rat:j:condorcet-inconsistencies,elk-lan-saf:c:condorcet-sets}.
On the other hand, our paper focuses on rules that, in some broader
sense, can be seen as extensions of positional scoring rules.
Our goal is to review some natural preference-based multiwinner rules,
to present a uniform framework for their study, and to propose a set
of natural properties (axioms) against which these rules can be
judged. In effect, we focus on the model where voters have ordinal
preferences.

We have picked ten voting rules as examples of different ideas
pertaining to (scoring-based) multiwinner elections: STV, SNTV,
$k$-Borda, Bloc, three variants of Chamberlin--Courant's
rule~\cite{cha-cou:j:cc,bet-sli-uhl:j:mon-cc,bou-lu:c:chamberlin-courant},
and three variants of Monroe's
rule~\cite{mon:j:monroe,bet-sli-uhl:j:mon-cc,sko-fal-sli:j:multiwinner}. STV
and SNTV are well-known rules that are used for parliamentary
elections in some countries;\footnote{For example, the upper house of
  the Parliament of Australia uses a variant of STV; a variant of SNTV
  is used, e.g., in Puerto Rico.} Bloc is a rule that asks voters to
specify their favorite committee of $k$ candidates and selects those
$k$ candidates that were nominated more often than others; $k$-Borda
picks $k$ alternatives with the highest Borda scores and is
representative of rules used for picking $k$ finalists in a
competition (indeed, Formula~1 racing and Eurovision song contest use
scoring rules very similar to Borda).  Chamberlin--Courant's rule and
Monroe's rule are examples of rules that, like STV, focus on
proportional representation, but are based on explicitly assigning a
committee member (a representative) to each voter.
We also consider two rules based on approximation algorithms, for
the Chamberlin--Courant's rule~\cite{bou-lu:c:chamberlin-courant} and for
the Monroe's rule~\cite{sko-fal-sli:j:multiwinner}. We consider them as voting rules in their own right. All these rules can
be seen as being loosely based, in some way, on single-winner scoring
protocols. 

We are interested in judging the selected multiwinner rules  with respect to their
applicability in the following settings:
\begin{description}[leftmargin=15pt]%
\fixlist
\item[Parliamentary Elections.] Voting rules for such elections should
  respect the ``one person, one vote'' principle. %
  This is reflected in the requirement
  that each elected member should represent, roughly, the same number
  of voters. Some such rules are based on electoral districts, i.e.,
  separate (possibly multiwinner) elections are held in different parts of
  the country, while others treat the whole country as a single
  constituency, and focus on proportional representation of different population groups.

\item[Shortlisting.] Consider a situation where a position is filled
  at a university. Each faculty member ranks applicants in order to
  create a short-list of those to be invited for an interview.  One of the
  important requirements in this case is that if
  some candidate is shortlisted when $k$ applicants are selected, then
  this candidate should also be shortlisted if the list were
  extended to $k+1$ applicants.

\item[Movie selection.] Based on rankings provided by different
  customer groups, an airline has to decide which (few) movies to
  offer on their long-distance flights. It is important that each
  passenger finds something satisfying.
  This task is similar to parliamentary elections, but without the
  need to worry that each movie would be watched by the same number of
  people.
  It is, however, quite different from shortlisting: If there are two
  similar candidates, then for shortlisting we should, typically, take
  either both or neither, whereas in the context of movie selection it
  makes sense to pick at most one of them. Skowron et
  al.~\cite{sko-fal-lan:c:multiwinner} expand this view of multiwinner
  elections and provide a number of other examples and applications.

\iffalse   
\item[Search Engines.] A search engine must select a fixed number of
  web pages to display, say 50, and rank them in the order of
  decreasing relevance to a user's query. It takes as input rankings
  from a relatively small number of (internal) rank functions (each
  using different criteria). The search engine then aggregates these
  rankings and selects web pages to display.  It is important to
  filter out spam, e.g., by using an aggregation rule that satisfies
  the Extended Condorcet Criterion
  \cite{dwo-kum-nao-siv:c:rank-aggregation}.
\fi
\end{description}

We study properties of voting rules that are important in the
above-listed settings. We introduce committee monotonicity, solid
coalitions property, consensus committee property, and unanimity, and
adapt the standard notions of monotonicity, homogeneity, and
consistency to the multiwinner framework. We discuss related
literature and compare our approaches in Section~\ref{sec:literature}.

Our paper is a preliminary attempt to give a formal framework for the
study of preference-based multiwinner rules. Thus we use the word
\emph{axiom} quite freely, without meaning that it should be a
\emph{normative requirement}.  Our work focuses on multiwinner rules
that are based on scoring protocols. Such rules are, in some deep
sense, very different from those that are based on the Condorcet
criterion (indeed, in the single-winner case, no scoring protocol is
Condorcet consistent; in effect, a multiwinner rule that degenerates to
the case of a scoring rule for single-member committees cannot be
Condorcet consistent in any natural way). It would be very interesting
to apply our framework to Condorcet-based rules. We leave this as
future work.

The paper is organized as follows. In Section~\ref{sec:prelim}
we introduce the basic terminology used in this paper; in
Section~\ref{sec:rules} we define the rules that we study and
put forward two ways of classifying them.  
In Section~\ref{sec:axioms}, we define
several properties of multiwinner rules
and in Sections~\ref{sec:com-mon}--\ref{sec:con-hom} we study
particular groups of these properties in detail.
We present related literature in Section~\ref{sec:literature} and
conclude  in Section~\ref{sec:conclusions}.

\section{Preliminaries}\label{sec:prelim}
\noindent
An {\em election} is a pair $E = (C,V)$, 
where $C = \{c_1, \ldots, c_m\}$ 
is a set of {\em candidates} and $V = (v_1, \ldots, v_n)$ is a sequence of
{\em voters}. Each voter is described by a {\em preference order},   
which is a ranking of the candidates from the most desirable one to the
least desirable one. We denote the position of a candidate $c\in C$ 
in the preference order of a voter $v\in V$ by $\pos_v(c)$.
If $V_1$ and $V_2$ are two sequences of voters over
the same candidate set $C$, then $V_1+V_2$ denotes the
concatenation of $V_1$ and $V_2$. If $V$ is a sequence of voters and
$t$ is an integer, then $tV$ denotes the concatenation of $t$ copies of
$V$. For $E_1 = (C,V_1)$ and $E_2 = (C,V_2)$, we write
$E_1+E_2$ to denote $(C,V_1+V_2)$, and for $E = (C,V)$ and a positive
integer $t$, we write $tE$ to denote $(C,tV)$.
For an integer $n$, we denote $\{1, \ldots, n\}$ by $[n]$.
A {\em multiwinner voting rule} $\calR$ is a function that given an
election $E = (C,V)$ and a positive integer $k$, $k \leq \|C\|$,
returns a set $\calR(E,k)$ of $k$-element subsets of $C$, which we
call \emph{committees}. That is, a rule returns a set of committees
that are tied-for-winning.  In practice, one would need to combine
such a rule with a tie-breaking mechanism but, for simplicity, we
mostly disregard this issue here.  Brams and Fishburn
\cite{bra-fis:b:polsci:voting-procedures} introduced \emph{choose-$k$}
rules but their definition stipulates that such a rule selects {\em at
  least} $k$ alternatives. Two early papers that focus on rules
selecting committees of size {\em exactly} $k$ are those of
Debord~\cite{deb:j:prudent} and of Felsenthal and
Maoz~\cite{fel-mao:j:norms}.
We use the same approach and when we
need to emphasize the size of the committee, we use the term
\emph{$k$-committee selection rules}.

Requiring multiwinner rules to pick committees of exactly a given size
is natural if, for example, the goal is to elect a parliament whose
size is fixed by the constitution. However, as a consequence, we are
sometimes forced to elect Pareto-dominated candidates (e.g., if all
voters unanimously rank the candidates in the same order and $k > 1$).
Alternatively, we could require $\calR(E,k)$ to return committees of
up-to-$k$ members.  The latter approach is also studied in the
literature (either explicitly or implicitly), but we adopt the former
one due to its simplicity and applicability in our settings of
interest. %

\section{Multiwinner Voting Rules}\label{sec:rules}
\noindent
We now provide definitions of the multiwinner rules that we study and
discuss two general ways of defining them.

\subsection{Common Multiwinner Rules}
\noindent
Many multiwinner rules rely on ideas from single-winner rules,
so let us review these first. Many single-winner rules calculate the scores of alternatives to decide which one is best. Here are some most popular ways
of computing candidate scores.
\begin{description}[leftmargin=15pt]\fixlist
\item[Plurality score.] The plurality score of a candidate $c$ is the number
  of voters that rank $c$ first.
\item[${\bf t}$-approval score.] Let $t$ be a positive
  integer. The $t$-approval score of a candidate $c$ is the number of
  voters that rank $c$ among the top $t$ positions.
\item[Borda score.] Let $v$ be a vote over
  a candidate set $C$. The Borda score of a candidate
  $c\in C$ in $v$ is $\|C\|-\pos_v(c)$.
The  Borda score of $c$ in an election $E = (C,V)$ is the
  sum of $c$'s Borda scores from all voters in $V$.  
\item[${\bf s}$-score.] The tree above types of score are special
  cases of general scoring protocols. Consider a setting with $m$
  candidates and score vector ${\bf s}=(\row sm)$, where $s_1\ge
  s_2\ge \ldots\ge s_m$. We define the ${\bf s}$-score of an
  alternative $a$ in a profile $(\row vn)$ as
\[
\text{sc}_{\bf s}(a)= \sum_{i=1}^n s_{\text{pos}_{v_i}(a)}.
\] 
It is immediate to see that plurality score is simply the $(1,0, \ldots,
0)$-score, $t$-approval score is
the $(\underbrace{1,\ldots,1}_t,0,\ldots,0)$-score, and Borda score is
the $(m-1, m-2, \ldots, 0)$-score.
\end{description}
\vspace{-1mm}

Given these definitions, we are ready to describe the multiwinner
rules that we focus on in this paper.  Let $E = (C,V)$ be an election
and let $k \in [\|C\|]$ be the size of the committee that we seek.
We assume the parallel-universes tie-breaking
\cite{con-rog-xia:c:mle}, i.e., our rules return all the committees
that could result from some breaking of the ties occurring
during the computation of the rule.\smallskip

\begin{description}
\item[Single Transferable Vote (STV).]
STV is a %
multistage elimination rule that works as follows. 
If there is a candidate $c$ whose Plurality score
is at least %
$q = \left\lfloor \frac{\|V\|}{k+1} \right\rfloor + 1$
(the so-called Droop quota), 
we do the following: (a) include $c$ in the winning committee, (b) delete $q$ votes
where $c$ is ranked first, and (c) remove $c$ from all the
remaining votes.  If each candidate's Plurality score
is less than $q$, a candidate with the lowest Plurality score is
deleted from all votes. %
(There are also many other variants of STV;
we point the reader to the work of Tideman and Richardson~\cite{tid-ric:j:stv} for details.)%

\item[Single Nontransferable Vote (SNTV).]
Under SNTV, we return the %
$k$ candidates with the highest Plurality scores (thus one can think
of SNTV as simply $k$-Plurality).

\item[Bloc.]
Under Bloc, we return the %
$k$ candidates with the highest
$k$-approval scores. %
\item[${\bf k}$-Borda.]
Under $k$-Borda, we return the %
$k$ candidates with the highest Borda
scores. %
Debord~\cite{deb:j:k-borda} provided an axiomatic
characterization of this rule.%
\item[Chamberlin--Courant's and Monroe's Rules.]
These rules explicitly aim at proportional representation. The main idea is to provide an optimal assignment of
committee members to voters by using a satisfaction function to
measure the quality of the assignment.

A {\em satisfaction function} is a 
nonincreasing mapping
$\alpha \colon \naturals \rightarrow \naturals$.
Intuitively, $\alpha(i)$ is a voter's satisfaction from
being represented by a candidate that this voter ranks in position $i$.
We focus on the Borda satisfaction function, which for $m$ candidates
is defined as $\alpha^m_\borda(i) = m-i$.

Let $k$ be the target size of the committee. 
A function $\Phi \colon V \rightarrow C$ is called an {\em assignment function} and it is called a $k$-{\em assignment function} if
$\|\Phi(V)\| \leq k$. Intuitively, %
$\Phi(V)$ is the elected committee where voter $v$ is
represented by candidate $\Phi(v)$. There are several ways to
compute the societal satisfaction from the assignment; we focus on the
following two:
\[
  \ell_1(\Phi) = \sum_{v \in V} \alpha( \pos_v(\Phi(v)) ), \qquad
  \ell_{\min}(\Phi) = \min_{v \in V}( \alpha(\pos_v(\Phi(v)) ),
\]
where $\alpha$ is the given satisfaction function.
The former one, $\ell_1(\Phi)$, is a utilitarian measure, which
sums the satisfactions of all the voters, and the latter one,
$\ell_{\min}(\Phi)$, is an egalitarian measure, which consider the
satisfaction of the least satisfied voter.

Let $\alpha$ be a satisfaction function and let $\ell$ be 
$\ell_1$ or $\ell_{\min}$. Chamberlin--Courant's rule with parameters $\ell$ and
$\alpha$ ($\ell$-$\alpha$-CC)  finds an assignment function
$\Phi$ that maximizes $\ell(\Phi)$ and declares  the
candidates in $\Phi(V)$ to be the winning committee. If $\|\Phi(V)\| < k$, the rule fills
in the missing committee members in an arbitrary way and outputs all resulting committees.
$\ell$-$\alpha$-Monroe's rule is defined
in the same way, %
except that we optimize over
assignment functions that additionally satisfy the
so-called Monroe criterion, which requires that 
$ \lfloor \frac{n}{k} \rfloor\leq \|\Phi^{-1}(c)\|\leq \lceil
\frac{n}{k} \rceil$ for each elected candidate $c$.  To simplify
notation, we omit $\alpha^m_\borda$ when referring to Monroe/CC rule
with the Borda satisfaction function.

For Chamberlin--Courant's rule, for each set of candidates $C'
\subseteq C$ we define the assignment function $\Phi^\cc(C')$ so that for
each voter $v$, $\Phi^\cc(C')(v)$ is $v$'s top candidate in $C'$.
If $W$ is a winning committee under
Chamberlin--Courant's rule, then $\Phi^\cc(W)$ is an optimal
assignment function.

The utilitarian variants of the rules (i.e., $\ell_1$-CC and
$\ell_1$-Monroe) were introduced by Chamberlin and
Courant~\cite{cha-cou:j:cc} and by Monroe~\cite{mon:j:monroe},
respectively. The egalitarian variants were introduced by Betzler et
al.~\cite{bet-sli-uhl:j:mon-cc}.  Unfortunately, these rules are hard
to compute, irrespective of tie-breaking, both for Borda
satisfaction function~\cite{bou-lu:c:chamberlin-courant,bet-sli-uhl:j:mon-cc}
and for various approval-based satisfaction
functions~\cite{pro-ros-zoh:j:proportional-representation,bet-sli-uhl:j:mon-cc}.
\item[Approximate Variants of ${\boldsymbol\ell_{\bf 1}}$-Monroe and
  ${\boldsymbol\ell_{\bf 1}}$-CC.]
Hardness results for $\ell_1$-CC and
$\ell_1$-Monroe inspired research on designing efficient
approximation algorithms for these
rules~\cite{bou-lu:c:chamberlin-courant,sko-fal-sli:j:multiwinner}.
Here, in the spirit of Caragiannis et
al.~\cite{car-kak-kar-pro:c:dodgson-acceptable}, we consider these
algorithms as full-fledged multiwinner rules.

We refer to the rules based on approximation algorithms for
$\ell_1$-CC and $\ell_1$-Monroe as Greedy-CC and Greedy-Monroe,
respectively. Greedy-CC was proposed by Lu and
Boutilier~\cite{bou-lu:c:chamberlin-courant} and Greedy-Monroe by
Skowron et al.~\cite{sko-fal-sli:j:multiwinner}. Both rules use the
Borda satisfaction function, aggregated in the utilitarian way, using
$\ell_1$. They proceed in $k$ iterations, in which they build sets
$\emptyset = W_0 \subset W_1 \subset W_2 \subset \cdots \subset W_k$,
and declare $W_k$ to be the winning committee. In the $i$-th
iteration, $i\in [k]$, Greedy-CC picks a candidate $c_i \in C
\setminus W_{i-1}$ that maximizes
$\ell_1(\Phi^\cc(W_{i-1}\cup\{c_i\}))$, where $W_i = W_{i-1} \cup
\{c_i\}$.  In particular, $c_1$ is an alternative with the highest
Borda score.

Greedy-Monroe, in addition to the sets $W_0, \ldots, W_k$,
also maintains sets of voters $\emptyset = V_0 \subset V_1 \subset
\cdots \subset V_k = V$, such that, after the $i$-th iteration, $V_i$
is the set of voters for which the rule has already assigned
candidates. In the $i$-th iteration, the rule picks a number $n_i \in
\{\lceil \frac{n}{k}\rceil, \lfloor \frac{n}{k}\rfloor\}$ (see below
for the choice criterion) and then picks a candidate $c_i \in C
\setminus W_{i-1}$ and a group $V'$ of $n_i$ voters 
from $V\setminus V_{i-1}$ that together maximize the Borda score of $c_i$ in
$V'$. The rule sets $W_i = W_{i-1}\cup \{c_i\}$ and $V_i = V_{i-1} \cup
V'$ (intuitively, Greedy-Monroe assigns $c_i$ to the voters in
$V'$). Regarding the choice of $n_i$, if $n$ is of the form $kn'+n''$,
where $0 \leq n'' < k$, then Greedy-Monroe picks $\lceil
\frac{n}{k}\rceil$ for the first $n''$ iterations and picks $\lfloor
\frac{n}{k}\rfloor$ for the remaining ones. In particular, $c_1$ is the Borda winner of a group $V_1$ of $n_1$ voters which maximizes the Borda score across all subsets of voters of size $n_1$.

Greedy-CC and Greedy-Monroe output committees that approximate those
output by $\ell_1$-CC and $\ell_1$-Monroe.  In particular, Greedy-CC finds a
committee $W$ such that the satisfaction of the voters is at least
$1-\frac{1}{e}$ of the satisfaction achieved under
$\ell_1$-CC~\cite{bou-lu:c:chamberlin-courant}, and Greedy-Monroe
finds a committee that achieves at least a $1-\frac{k}{2m-1} - \frac{H_k}{k}$ fraction of the satisfaction given by
$\ell_1$-Monroe, where $H_k=\sum_{i=1}^{k}\frac{1}{k}$~\cite{sko-fal-sli:j:multiwinner} (for many practical setting, this value is quite close to $1$). 
These rules are efficiently computable in the sense that we can output 
some winning committee in polynomial time; however, their computational complexity under
parallel-universes tie-breaking is not known.
\end{description}

\noindent
\noindent\textbf{Hardness of Winner Determination.}\quad
For $\ell_1$-Monroe and $\ell_1$-CC,  their hardness
of winner determination is known~\cite{pro-ros-zoh:j:proportional-representation,bou-lu:c:chamberlin-courant,bet-sli-uhl:j:mon-cc,sko-fal-sli:j:multiwinner}. Other
rules that we study, on the surface, are polynomial-time
computable. However, since in this paper we very broadly apply the
parallel-universes tie-breaking, this polynomial-time complexity is
not necessarily obvious. For SNTV, Bloc, and $k$-Borda, even with this
tie-breaking, winner determination is computationally easy. On the
other hand, for STV the problem is
$\np$-hard~\cite{con-rog-xia:c:mle}. For Greedy-CC and Greedy-Monroe
the answer is currently unclear, but for both of these rules there are
simple tie-breaking mechanisms under which they are polynomial-time
computable and which do not break any of the axiomatic properties that
we study.

\subsection{Two Types of Multiwinner Rules}\label{sec:two-types}
\noindent
Perhaps surprisingly, it turns out that many of the 
rules introduced so far have very similar internal structure. Below we present two
natural ways of identifying these similarities.\smallskip

\noindent\textbf{Best-${\bf k}$ Rules.}\quad 
SNTV and $k$-Borda are natural extensions of Plurality and Borda to
the multiwinner setting: We sort the candidates in the order of
decreasing scores (with further parallel-universes tie-breaking if needed) 
and pick the top $k$ ones. \par

We remind the reader that a {\em social preference function} $F$ is a
function that, given an election $E = (C,V)$, returns a set of tied
linear orders over $C$ and a social welfare function returns just one
but maybe non-strict linear order.  Breaking ties in the latter using
the parallel-universes tie-breaking we can also convert a social
welfare function into a social preference function (in effect, we will
often treat social welfare functions as special cases of social
preference functions).

\begin{definition}
We say that a given multiwinner rule $\calR$ is a \emph{best-$k$ rule} if there is a
social preference function $F$ such that for each $k \in [m]$, 
a set $W$ is in $\calR(E, k)$ if and only if $\|W\|=k$ and 
there is an order $\mathord{\succ}$ in
$F(E)$ such that $c \succ d$
for each $c \in W$ and $d \in C \setminus W$.
\end{definition}

SNTV and $k$-Borda are best-$k$ rules.  Indeed, a class of best-$k$ rules can be defined by providing a score vector 
${\bf s}=(\row sm)$, where $s_1\ge s_2\ge \ldots\ge s_m$. %
This defines a social welfare function which ranks all alternatives in
accord with their ${\bf s}$-scores, which gives us a best-$k$
rule. As we will see later, perhaps unexpectedly, Bloc is not a
best-$k$ rule.

We can also define a best-$k$ rule based on the social preference function known as
the Kemeny ranking~\cite{kem:j:no-numbers}, and, somewhat surprisingly, we will later
note that Greedy-CC is also a best-$k$ rule. 
Thus, best-$k$ rules are a more diverse group than one
might at first expect.\smallskip %

\noindent\textbf{Committee Scoring Rules.}\quad
Both $k$-Borda and $\ell_1$-CC can be viewed as generalizations of
the Borda rule to the multiwinner case. Here we introduce a class of
{\em committee scoring rules}, which generalize single-winner
scoring rules capturing $k$-Borda, $\ell_1$-CC, and many
other rules. We believe that identifying committee scoring rules is
an important conceptual contribution of this paper.

Consider an election $E = (C,V)$ where we want to pick a committee of
size $k$ out of $m = \|C\|$ candidates. A $k$-winner committee scoring
rule is defined via a {\em committee scoring function} $f$, $f\colon
[m]^k \rightarrow \naturals$, as follows. Given a committee
$S$ and a voter $v$, we define $\pos_v(S)$ to be the vector $(i_1,
\ldots, i_k)$ resulting from sorting the set $\{\pos_v(c) \mid c \in
S\}$ in the nondecreasing order. The winning committees
are the ones that maximize the sum:
\[
\text{sc}_f(S)=\sum_{v \in V}f(\pos_v(S)),
\]
which we call the {\em score} of the committee $S$ under $f$.

Just as for single-winner scoring rules, we require a certain form of
monotonicity with respect to the values of $f$. Specifically, 
if $I = (\row ik)$ and $J = (\row jk)$ are two increasing
sequences of numbers from $[m]$, we say that $I \succeq J$ if and
only if $(i_1\le j_1)\land \ldots \land (i_k\le j_k)$, and we require that 
$I \succeq J$ implies $f(I) \geq f(J)$.

\begin{example}
\label{ex1}
  Let $m$ be the number of candidates and let $k$ be the target size of the committee.
  Define $\alpha_\ell\colon [m] \rightarrow
  \{0, 1\}$ by setting $\alpha_\ell(i) = 1$ if $i \leq \ell$ 
  and $\alpha_\ell(i) = 0$ otherwise. For each
  $i\in [m]$, define $\beta(i) = m-i$.
  SNTV, Bloc, $k$-Borda, and $\ell_1$-CC are committee scoring rules
  defined by:
  \begin{align*}
    f_\sntv(\row ik) &=%
    \alpha_1(i_1),  \\
    f_\bloc(\row ik) &= \textstyle\sum_{t=1}^{k}\alpha_k(i_t), \\
    f_\kborda(\row ik) &= \textstyle\sum_{t=1}^k\beta(i_t), %
    \\
    f_\cc(\row ik) &= \beta(i_1).
  \end{align*}
\end{example}

\noindent
We note that in the case of $f_\sntv$ and $f_\kborda$, the function
$f$ has the form:
 \begin{equation}
 \label{alphas}
 f(\row ik)=\gamma(i_1)+\ldots+\gamma(i_k),
 \end{equation}
 for some non-increasing function $\gamma $ not depending on $k$. This
 is immediate for $f_\kborda$ and for $f_\sntv$ it follows by the
 nature of $\alpha_1$. On the other hand, $f_\cc$ is clearly not of
 the given form, and neither is $f_\bloc$ because function $\alpha_k$
 depends on $k$.  We refer to committee scoring functions of the
 form~\eqref{alphas} as \emph{separable}.  

 Let us focus on a separable committee scoring rule with a scoring
 function of the form~\eqref{alphas}, and let us denote
 $\gamma(i)=s_{i}$, $i=1,\ldots,m$. We form a scoring vector
 $\mathbf{s} = (s_1, \ldots s_m)$; since $\gamma$ is non-increasing,
 we have $s_1\ge s_2\ge\cdots\ge s_m$ and, so, $\mathbf{s}$ is indeed
 a valid scoring vector.  Now, given some committee $S=\{\row ak\}$,
 it is easy to verify that the total score of $S$ at some profile $V =
 (\row vn)$ is:
\[
\sum_{i=1}^n\sum_{j=1}^k s_{\text{pos}(a_j,v_i)}=\sum_{i=1}^k \text{sc}_{\bf s}(a_i).
\]
Hence, the voting rule is the best-$k$ rule induced by the scoring
social welfare function defined by the score vector ${\bf s}$. In effect,
we have the following theorem.

 \begin{theorem}
 \label{separable-kbest}
 Every separable $k$-committee scoring rule is a $k$-best rule for
 some scoring social welfare
 function.%
 \end{theorem}

 According to our definition, Bloc is not a separable committee
 scoring rule and, in some ways, it is quite distinct from separable
 rules (e.g., it is not a best-$k$ rule; we give another example that
 Bloc stands out at the end of Section~\ref{sec:axioms}). On the other
 hand, at the formal level Bloc does look very similar to separable
 committee scoring rules and it shares a number of features with
 them. To capture this fact, we introduce the notion of a weakly
 separable committee scoring rule.

 \begin{definition}\label{def:weakly}
   We say that a committee scoring rule $\calR$, defined through
   scoring function $f$, is weakly separable if there exists a family
   $\{ \alpha^m_k \mid m\in \naturals, k \leq m, \alpha^m_k\colon [m]
   \rightarrow \naturals \}$ of nonincreasing functions such that for
   each $m \in \naturals$, $k \leq m$, and sequence $0 < i_1 < \cdots
   < i_k \leq m$, we have $f(\row ik) = \sum_{t=1}^k\alpha_k^m(i_t)$.
 \end{definition}

 The main difference between separable committee scoring rules and
 weakly separable ones is that weakly separable ones allow function
 $f$ to use a score vector that depends on $k$ (as in the case of
 Bloc).  

 Weakly separable committee scoring rules form an interesting class of
 multiwinner rules. For example, It is immediate to see that they are
 polynomial-time computable (provided that they are defined through a
 family of polynomial-time computable scoring function) and in
 Section~\ref{sec:monotonicity} we show that they have an interesting
 monotonicity property.
 
\begin{proposition}\label{pro:sep}
  If $\calR$ is a weakly separable committee scoring rule defined
  through a polynomial-time computable family of functions (in the
  sense of Definition~\ref{def:weakly}), then $\calR$ is polynomial-time
  computable.
\end{proposition}

On the other extreme, we have committee scoring rules defined through
functions $f(\row ik)$ whose values depend solely on $i_1$ (like the fourth rule in Example~\ref{ex1}). Such rules
seem to focus on voter representation. As in
$\ell_1$-CC, each voter is assigned to her most preferred
candidate in the selected committee, and only contributes
towards the score of this candidate. We call such committee scoring
rules \emph{representation-focused}.

It seems that representation-focused rules, in essence variants of the
Chamberlin--Courant rule, are $\np$-hard to
compute~\cite{pro-ros-zoh:j:proportional-representation,bou-lu:c:chamberlin-courant,bet-sli-uhl:j:mon-cc},
but there is at least one exception: We note that SNTV is both
separable and representation-focused and, in effect, is
polynomial-time computable. One could say that this is because of its
separability. While this indeed is a convincing explanation, there
also are non-separable committee scoring rules that are
polynomial-time computable, with Bloc being perhaps the simplest
example (though, Bloc is weakly separable).

\section{Axioms}\label{sec:axioms}
\noindent
We now put forward some properties (axioms) that multiwinner rules may
or may not satisfy. We use the standard axioms for single-winner rules
as our starting point, and augment them with ideas from the literature
that are specific to the multiwinner domain.
Due to our choice of focus, we do not include properties based on
the Condorcet principle, such as, e.g.,  the stability of Barber\`a and
Coelho~\cite{bar-coe:j:non-controversial-k-names}.
We stress that, since multiwinner rules have a very diverse range of applications,
our properties should not necessarily be understood in the normative
way:  %
the desirability of a
particular property can only be evaluated in the context of a specific
application.
Throughout this section, we write $\calR$ to denote some given
multiwinner rule.

\medskip

\begin{table*}
\begin{center}
\small
  \newcommand{\yes}{$\surd$}
  \newcommand{\no}{$\times$}
  \newcommand{\full}{full}
  \newcommand{\noncross}{NC}
  \newcommand{\cand}{C}
  \newcommand{\cnc}{C/NC}
  \begin{center}
  \tabcolsep=0.15cm
  \begin{tabular}{l|cccccccc}
    & Committee & Solid & Consensus && Monoto-& Homoge-& Consis- & \\
    Rule & 
    Monotonicity &
    Coalitions &
    Committee &
    Unanimity &
    nicity &
    neity &
    tency 
    \\
    \hline
    STV                 & \no  & \yes\ ($\diamond$) & \yes\ ($\diamond$) & strong &\no & \yes ($\heartsuit$)& \no\\[1pt]
    SNTV                & \yes & \yes & \yes & weak   & \cnc & \yes & \yes\\[1pt]
    Bloc                & \no  & \no  & \no  & fix maj. & \cnc & \yes & \yes\\[1pt]
    $k$-Borda           & \yes & \no  & \no  & strong & \cnc & \yes & \yes\\[1pt]
    $\ell_1$-CC         & \no  & \no  & \yes & weak   & \cand & \yes              & \yes\\[1pt]
    $\ell_{\min}$-CC     & \no  & \no  & \yes & weak   & \cand & \yes              & \no\\[1pt] 
    Greedy-CC           & \yes & \no  & \no  & weak   & \no  & \yes               & \no\\[1pt]
    $\ell_1$-Monroe     & \no  & \no  & \yes & strong & \no  & \yes\ ($\clubsuit$)  & \no\\[1pt]
    $\ell_{\min}$-Monroe & \no  & \no  & \yes & strong & \no  & \yes\ ($\clubsuit$)  & \no\\[1pt]
    Greedy-Monroe       & \no  & \yes & \yes & strong & \no  & \multicolumn{1}{c}{\yes\ ($\spadesuit$)}  & \no
  \end{tabular}
  \end{center}
  \caption{\label{tab:results}Summary of results.
 $\boldsymbol{\surd}$ and $\boldsymbol{\times}$ indicate that 
    the rule has/does not have the respective property.
    C means candidate monotonicity and NC means
    non-crossing monotonicity (C/NC means satisfying both conditions). 
    ``Fix maj.'' in the Unanimity column for Bloc means that not only
    Bloc is unanimous in the strong sense, but also that it satisfies
    the fixed majority property, which is stronger (none of the other
    rules satisfy it).
    The properties marked with 
    ($\diamond$) hold for STV when $n \geq k(k+1)$; property marked with 
    (${\heartsuit}$) requires STV to use non-rounded Droop quota and fractional votes. Properties 
    marked with (${\clubsuit}$) hold if ${n}$ is divisible by ${k}$ 
    and $({\spadesuit}$) in addition requires a specific intermediate tie-breaking rule.}

\end{center}
\end{table*}

Our first axiom is {\em nonimposition}. It requires that
each size-$k$ set of candidates can win. This is a basic
requirement that is trivially satisfied by all rules that we
consider.

\begin{description}[leftmargin=15pt]
\item[Nonimposition.] For each set of candidates $C$ and each
  $k$-element subset $W$ of $C$, there is an election $E = (C,V)$ such
  that $\calR(E,k) = \{W\}$.
\end{description}
The next three axioms---consistency, homogeneity, and monotonicity---are adapted 
from the single-winner setting. 
For the first two, the adaptation is straightforward.

\begin{description}[leftmargin=15pt]\fixlist

\item[Consistency.] For every pair of elections $E_1 = (C,V_1)$, $E_2 =
  (C,V_2)$ and each $k \in [\|C\|]$, if $\calR(E_1,k) \cap
  \calR(E_2,k) \neq \emptyset$ then $\calR(E_1+E_2,k) = \calR(E_1,k)
  \cap \calR(E_2,k)$.

\item[Homogeneity.] For each election $E = (C,V)$, each $k \in
  [\|C\|]$, and each $t \in \naturals$, %
  it holds that $\calR(tE,k) = \calR(E,k)$.
\end{description}
We now consider monotonicity. If $c$ belongs to a winning committee
$W$ then, generally speaking, we cannot expect $W$ to remain winning
when $c$ is moved forward in some vote.  For example, this shift may
hurt other members of $W$.  Indeed, none of our rules satisfies this
strict version of monotonicity. However, there are two natural
relaxations of this condition. One option is to require that after the
shift $c$ belongs to \emph{some} winning committee. Alternatively, we
may restrict forward movements of $c$, prohibiting it to overtake
other members of~$W$. (We point the reader to the work of Sanver and
Zwicker~\cite{san-zwi:j:monotonicity} for an extensive discussion of
monotonicity in the context of irresolute voting rules.)

\begin{description}[leftmargin=15pt]
\item[Monotonicity.] For each election $E = (C,V)$, each $c \in C$,
  and each $k \in [\|C\|]$, if $c\in W$ for some $W \in \calR(E,k)$,
  then for each $E'$ obtained from $E$ by
  shifting $c$ one position forward in some vote $v$ it holds that:
(1) for {\em candidate monotonicity}:  $c\in W'$ for some $W' \in \calR(E',k)$, and
(2) for {\em non-crossing monotonicity}: if $c$ was ranked immediately below some $b\not\in W$, then $W \in \calR(E',k)$.
\end{description}

Our next axiom, {\em committee monotonicity}, is specific to multiwinner elections,
as it deals with changing the size of the desired committee.  Intuitively, it
requires that when we increase the size of the target committee, none of
the already selected candidates should be dropped. 
Our phrasing is somewhat involved because $\calR$
returns sets of committees.

\begin{description}[leftmargin=15pt]
\item[Committee  Monotonicity.] %
  For each election $E = (C,V)$ the following conditions hold: (1) For each
  $k \in [m-1]$, if $W \in \calR(E,k)$ then there exists a $W' \in
  \calR(E,k+1)$ such that $W \subseteq W'$; (2) for each $k \in
  [m-1]$, if $W \in \calR(E,k+1)$ then there exists a $W' \in
  \calR(E,k)$ such that $W' \subseteq W$.
\end{description}
The second condition in the definition above is aimed to prevent the
following situation. Consider an election $E$ with candidate
set $C = \{a,b,c, \ldots\}$. %
Without condition~(2) a
committee-monotone rule $\calR$ would be allowed to output
the following winning committees: 
$\calR(E,1) = \{\{a\}\}$, $\calR(E,2) = \{\{a,b\}, \{b,c\}\}$,
and so on. Note that the committee $\{b, c\}$, %
which suddenly appears
in $\calR(E,2)$, %
breaks what we would intuitively think of as committee
monotonicity, but is not ruled out by condition (1) alone. 

Other authors already discussed committee monotonicity in the context
of multiwinner voting, but under different names and usually by
arguing that it is paradoxical that a given voting rule fails
committee
monotonicity~\cite{sta:j:paradoxes,rat:j:condorcet-inconsistencies}. In
Section~\ref{sec:com-mon} we discuss committee monotonicity in detail,
argue that failing it may be seen as a positive feature of a
multiwinner rule, and use it to axiomatize best-$k$ rules.

The final three axioms represent three implementations of Dummett's
condition known as {\em proportionality for solid
  coalitions}~\cite{dum:b:voting}.  Dummett's original proposal is as
follows: Consider an election with $n$ voters where the goal is to
pick $k$ candidates. If for some $\ell \in [k]$ there is a group of
$\frac{\ell n}{k}$ voters that all rank the same $\ell $ candidates on
top, these $\ell $ candidates should be in a winning committee. This
requirement, which tries to capture the idea of proportional
representation, seems to be quite strong. We are not aware of a single
rule that satisfies it.\footnote{There is also a variant of Dummett's
  condition known as \emph{Droop Proportionality Criterion}, which is
  geared toward STV~\cite{woo:j:properties}; STV can be shown to
  satisfy Dummett's condition whenever Droop quota is smaller than
  $\frac{n}{k}$.}  The following three axioms are weaker and reflect
the same idea.

\begin{description}[leftmargin=15pt]\fixlist
\item[Solid Coalitions.] For each election $E = (C,V)$ and each $k \in
  [\|C\|]$, if at least $\frac{\|V\|}{k}$ voters rank some candidate
  $c$ first then $c$ belongs to every committee in $\calR(E,k)$.

\item[Consensus Committee.] For each election $E = (C,V)$ and each $k
  \in [\|C\|]$, if there is a $k$-element set $W$, $W \subseteq C$,
  such that each voter ranks some member of $W$ first and each member
  of $W$ is ranked first by either $\lfloor \frac{\|V\|}{k} \rfloor$
  or $\lceil \frac{\|V\|}{k} \rceil$ voters then $\calR(E,k) = \{W\}$.

\item[Unanimity.] For each election $E =
  (C,V)$ and each  $k \in [\|C\|]$, if each voter ranks the same $k$ candidates
  $W$ on top (possibly in different order), then $\calR(E,k) =
  \{W\}$ (strong unanimity) or $W \in \calR(E,k)$ (weak unanimity).
\end{description}

We remind the reader that we list only the axioms that make sense for
preference-based rules which are in some broad sense close to scoring
rules. There are however preference-based axioms that are geared
towards Condorcet principle. The following axiom is an
example~\cite{deb:j:prudent}, though it can also be seen as a
generalization of the unanimity property.

\begin{description}
  \item[Fixed Majority.] For each election $E = (C,V)$ and each
  $k \in [\|C\|]$, if there is a $k$-element set $W$,
  $W \subseteq C$, such that a strict majority of voters rank all member of $W$  above all non-members of $W$, then $\calR(E,k) = \{W\}$.
\end{description}

Almost all the rules that we consider in this paper fail to satisfy
this axiom. For most of them, this already happens in the
single-winner setting, with $k=1$. Indeed, Plurality is the only
single-winner scoring protocol that guarantees that a candidate ranked
on top by a majority of the voters is the unique winner.  However,
quite interestingly, Bloc does satisfy the fixed majority
property.\footnote{ To see that this is the case, let $W$ be a
  committee such that a majority of voters rank all the members of $W$
  on top $k$ positions, and let $S$ be some other committee. It is easy
  to see that, under Bloc, $W$ has a higher score than $S$. Both
  committees receive the same score for the candidates in $W \cap S$,
  but the candidates in $W \setminus S$ receive strictly more points
  than the candidates in $S \setminus W$.} This is yet another reason
to view Bloc as a rule of a different kind than separable committee
scoring rules (even though at the level of formal definition Bloc is
very close to separable rules; see the discussion in
Section~\ref{sec:two-types}).

\section{Committee Monotonicity}\label{sec:committee-monotonicity}\label{sec:com-mon}
\noindent
The desirability of committee monotonicity depends strongly on the
application: if we are choosing finalists of a competition, then it is
imperative to use a rule that has this property, but in the context of
proportional representation requiring that the rule is committee
monotonic may prevent selecting a truly representative committee.
Indeed, this was already suggested by
Black~\cite{bla:b:polsci:committees-elections}: Consider a society
with single-peaked preferences regarding the left-right political
spectrum. If we are picking a single candidate (i.e., if $k=1$) then
it is most natural to select as centrist a candidate as possible
(formally, we would select a Condorcet winner). However, if we were to
select two candidates to represent the society (i.e., if $k=2$), then
selecting a ``moderate left-wing'' and a ``moderate right-wing''
candidate would, intuitively, give more proportional representation
than selecting the centrist candidate and an additional one (if the
additional candidate were to the left of the centrist one, the
right-wing voters would be neglected; if the candidate were to the
right of the centrist one, the left-wing voters would be neglected).

The above intuitions are further strengthened by the fact that
committee monotonicity axiomatically characterizes the class of
best-$k$ rules.

\begin{theorem}\label{thm:com-mon}
A $k$-committee selection rule is committee
  monotonic if and only if this rule is a best-$k$ rule.
\end{theorem}
\begin{proof}
  Let $\calR$ be a best-$k$ rule and let $F$ be the underlying social
  preference function. Consider an election $E = (C,V)$.  
  Pick $k\in [\|C\|-1]$ and   $W\in \calR(E,k)$. 
  By definition of a $k$-best rule, there
  is an order $\mathord{\succ}$ in $F(E)$ such that $w \succ c$ for each $w \in
  W$ and each $c \in C \setminus W$.
  Clearly, there is a candidate $w' \in C \setminus W$ such that for
  each $w \in W$ and each $c \in C\setminus (W\cup\{w'\})$ we have $w
  \succ w' \succ c$. Hence, $W \cup \{w'\} \in
  \calR(E,k+1)$. A similar argument shows that $\calR$ satisfies
  the second committee-monotonicity condition.
  Conversely, 
  assume that $\calR$ satisfies committee monotonicity.
  We will show that it is a best-$k$ rule by revealing the underlying
  social preference function $F$. Let $E = (C,V)$ be some election where
  $C = \{c_1, \ldots, c_m\}$. We define $F(E)$ to contain all linear
  orders $\mathord{\succ}$ that satisfy the following condition: If
  $\pi$ is a permutation of $[m]$ and $c_{\pi(1)} \succ
  c_{\pi(2)} \succ \cdots \succ c_{\pi(m)}$ then there is a sequence
  of sets $W_1 = \{c_{\pi(1)}\}, W_2 = \{c_{\pi(1)},c_{\pi(2)}\},
  \ldots, W_m = \{c_{\pi(1)}, \ldots, c_{\pi(m)}\}$ such that $W_1 \in
  \calR(E,1), W_2 \in \calR(E,2), \ldots, W_m \in \calR(E,m)$.  Using
  the two conditions from the definition of committee monotonicity, it
  is easy to verify that $F$ indeed defines $\calR$.
\end{proof}

Thus SNTV, $k$-Borda, and all  separable committee
scoring rules satisfy committee monotonicity.
On the other hand, Greedy-CC satisfies committee monotonicity and,
in effect, is a best-$k$ rule as well.

\begin{proposition}
  Greedy-CC satisfies committee monotonicity.
\end{proposition}
\begin{proof}
  Given an election $E$ and an integer $k$, Greedy-CC performs $k$
  iterations, in each picking one member of the committee.  However,
  the member of the committee picked in each iteration depends only on
  which members were chosen previously and not on $k$. Thus, after
  Greedy-CC performs $k$ iterations to compute committee of size $k$,
  one can simply perform one more iteration to obtain committee of
  size $k+1$. Conversely, if Greedy-CC is to compute a committee of
  size $k+1$, it first computes a committee of size $k$.  In effect,
  it is clear that Greedy-CC is committee monotone.
\end{proof}

For separable committee scoring rules, by
Theorem~\ref{separable-kbest} their underlying social welfare
functions are the scoring ones. This possibly distinguishes them from
rules such as Greedy-CC, whose underlying social preference functions
may not be based on scoring rules.

\begin{proposition}
  STV, Bloc, $\ell_1$-CC, $\ell_{\min}$-CC,
  $\ell_1$-Monroe, $\ell_{\min}$-Monroe, and Greedy-Monroe do not satisfy
  committee monotonicity.
\end{proposition}
\begin{proof}
  Let us first consider $\ell_1$-CC, $\ell_{\min}$-CC,
  $\ell_1$-Monroe, $\ell_{\min}$-Monroe, and Greedy-Monroe. Let $E =
  (C,V)$ be an election with $C = \{a,b,c,d\}$ and the following four
  votes: 
  \begin{align*}
     v_1\colon&  a\succ c\succ b\succ d, 
     &v_2\colon&  b\succ c\succ a\succ d,\\
     v_3\colon&  a\succ c\succ d\succ b,
     &v_4\colon& b\succ c\succ d\succ a.
  \end{align*}  
  It is easy to see that for $k=1$ the unique winning committee is
  $\{c\}$ but for $k=2$ the unique winning committee is $\{a,b\}$.

  The case of Bloc was resolved by Staring~\cite{sta:j:paradoxes}, but
  we include a simple argument for the sake of completeness.  We use
  election $E = (C,V)$ with $C = \{a,b,c\}$ and four votes:
  \newcommand{\vote}[3]{#1 \succ #2 \succ #3}
  \begin{align*}
     v_1\colon&  \vote abc, 
     &v_2\colon& \vote bca,\\
     v_3\colon&  \vote acb,
     &v_4\colon& \vote cba.
  \end{align*}  
  For $k=1$ Bloc has unique winning committee $\{a\}$, but with $k=2$ the\
  unique winning committee is $\{b,c\}$.

  Our example for STV is a bit more involved. Consider $E = (C,V)$
  with $C = \{a,b,c,d\}$ and with $24$ voters. We have:
  \begin{enumerate}\fixlist
  \item[(1)] $11$ voters with preference order $\fourvote abcd$, 
  \item[(2)] $3$ voters  with preference order $\fourvote bcad$, 
  \item[(3)] $4$ voters  with preference order $\fourvote cdab$, and
  \item[(4)] $6$ voters  with preference order $\fourvote dcab$.
  \end{enumerate}
  For $k=1$, the Droop quota is $\lfloor \frac{24}{2}\rfloor +1 = 13$. In
  the first round no candidate meets the quota and so STV eliminates
  the candidate with the lowest plurality score, that is, $b$. In the
  next round still no candidate meets the quota and STV eliminates
  $d$. In the resulting profile $c$ has plurality score $13$ and is the unique winner
  for $k=1$.

  For $k=2$, the Droop quota is $\lfloor \frac{24}{3}\rfloor+1 = 9$. Thus
  in the first round STV picks $a$ and removes it from the election
  together with $9$ voters that rank it first. The remaining two
  voters that supported $a$ transfer their votes to $b$, who now has
  plurality score $5$. Thus in the next round no candidate meets the
  quota and STV eliminates $c$. In the next round $d$ has plurality
  score $10$ and is selected. The unique winning committee is
  $\{a,d\}$.
\end{proof}

\begin{corollary}
STV, Bloc, $\ell_1$-CC, $\ell_{\min}$-CC, $\ell_1$-Monroe,
$\ell_{\min}$-Monroe, and Greedy-Monroe are not best-$k$ rules.
\end{corollary}

\section{Dummett's Proportionality}\label{sec:fairness}
\noindent
Properties in the spirit of Dummett's proportionality condition (with the
exception of unanimity) are geared toward rules that aim to achieve
proportional representation of the voters. Thus, in this section, we
judge multiwinner rules from this perspective.

We start by considering the solid coalitions
property.  It is easy to see that it is satisfied by both SNTV and STV. 

\begin{theorem}
  SNTV has the solid coalitions property. STV also has it for each election with $n \geq k(k+1)$, where $n$ is the number of voters and $k$ is the size of the winning committee.
\end{theorem}
\begin{proof}
  For SNTV it suffices to note that if there are $n$ voters and some
  candidate $c$ is ranked first by $\frac{n}{k}$  voters,
  then---by a simple counting argument---there cannot be $k$ other
  candidates  each of whom is ranked first by at least $\frac{n}{k}$
  voters. Thus $c$ must be included in each winning committee.

  For STV, %
  the Droop quota has value $\lfloor
  \frac{n}{k+1}\rfloor+1$ and, if $n \geq k(k+1)$, then it is smaller
  or equal to $\frac{n}{k}$. Thus if there is a candidate $c$ who is
  ranked first by $\frac{n}{k}$ candidates then this candidate will be
  included in the winning committee. %
\end{proof}
On the other hand, even though this property seems
to be very much in spirit of the Monroe and Chamberlin--Courant rules,
$\ell_1$-Monroe, $\ell_1$-CC, $\ell_{\min}$-Monroe, and
$\ell_{\min}$-CC fail to satisfy it. %
Yet, it is satisfied by Greedy-Monroe.

\begin{theorem}
  $\ell_1$-CC, $\ell_{\min}$-CC, $\ell_1$-Monroe, and $\ell_{\min}$-Monroe
  do not have the solid coalitions property, but Greedy-Monroe does have it.%
\end{theorem}
\begin{proof}
Let us consider an election with candidate set $C =
  \{a,b,c,d,e,f,g\}$ and nine voters whose preference orders are 
  \newcommand{\sixvote}[7]{#1 \succ #2 \succ #3 \succ #4 \succ #5 \succ #6 \succ #7}
  \begin{align*}
    v_1 \colon& \sixvote aedfgbc,\\
    v_2 \colon& \sixvote bfdegac,\\
    v_3 \colon& \sixvote cgdefab,\\[7pt]
    v_4 \colon& \sixvote aedfgbc,\\
    v_5 \colon& \sixvote bfdegac,\\
    v_6 \colon& \sixvote cgdefab,\\[7pt]
    v_7 \colon& \sixvote daefgbc,\\
    v_8 \colon& \sixvote dbefgac,\\
    v_9 \colon& \sixvote dcefgab.
  \end{align*}
  It can be easily verified that none of the four versions of the
  Chamberlin-Courant and Monroe rules elect a committee of size~$3$
  that contains $d$, even though this would be required by the solid
  coalitions property. Indeed, for the committee $\{a,b,c\}$ the total
  satisfaction is $9\|C\|-12$ in the utilitarian version and
  $\|C\|-1$ in the egalitarian one. However if $d$ is on the
  committee, then at most two of the candidates among $a,b,c,e, f, g$
  are also in the committee and so the total satisfaction is,
  respectively, at most $9\|C\|-13$ and
  $\|C\|-2$. %

  For the second part of the theorem, we consider Greedy-Monroe. Take
  some election with $n$ voters, where we seek a committee of size
  $k$.  Suppose that some candidate $c$ is ranked first by at least
  $\frac{n}{k}$ voters. Greedy-Monroe starts by picking candidates
  ranked first by at least $\frac{n}{k}$ voters.  By the time it
  considers $c$, each of the voters that rank $c$ first remains
  unassigned, so it picks $c$.
\end{proof}

We believe that the solid coalitions property is desirable,  but not crucial for applications that require proportional representation (e.g., parliamentary elections). 
In contrast, the consensus committee property, which
we discuss next, seems to be fundamental. Indeed, it 
is satisfied by almost all rules that %
aim to achieve proportional representation.

When $k$ is a divisor of $n$, the consensus committee property is satisfied
by every rule that has the solid coalitions property. %
In particular, it is satisfied by SNTV, STV (if there are
sufficiently many voters) and Greedy-Monroe. It is also 
satisfied by $\ell_1$-CC, $\ell_{\min}$-CC, $\ell_1$-Monroe, and
$\ell_{\min}$-Monroe, but, interestingly, not by Greedy-CC.
This reveals a major deficiency of the latter rule: It makes decisions
regarding the inclusion of some candidate $c$ into the committee based
on the preferences of the voters to whom $c$ would not
be assigned. This is very problematic for a rule that seeks to
approximate $\ell_1$-CC; interestingly, the feature of Greedy-CC that
causes this is also responsible for the rule being committee monotone.

\begin{proposition}\label{prop-no-consensus-committee}
  Bloc, $k$-Borda and Greedy-CC do not have the consensus committee
  property (nor the solid coalitions property).
\end{proposition}

\begin{proof}
  Consider an election with $C = \{a,b,c,d\}$ and two
  voters with preference orders $b \succ c \succ d \succ a$ and $a
  \succ c \succ d \succ b$. We seek a committee of size $k = 2$.  Then the consensus committee is $\{a,b\}$ but each
  of these rules  includes $c$ in each winning
  committee and thus fails the consensus committee property.
\end{proof}

For SNTV, $\ell_1$-CC, $k$-Borda and Bloc, the above
results can also be seen as incarnations of the following two more
general results regarding committee scoring rules.

\begin{proposition}
  Let $\calR$ be a separable committee scoring rule, let $k < m$, and
  let $f(i_1, \ldots, i_k) = \sum_{t=1}^k\gamma(t)$ be the respective
  committee scoring function. Then $\calR$ fails the consensus
  committee property if $0 < \gamma(1) \leq k\gamma(2)$, but satisfies
  it if $\gamma(1) > k\gamma(2)$ and there are sufficiently many
  voters.
\end{proposition}
\begin{proof}
We first
  consider the case when $0 < \gamma(1) \leq k\gamma(2)$. Consider an
  election $E = (C,V)$ with $m$ candidates and %
  $k+1$ voters. Let $c_1, \ldots, c_k,x$ be some $k+1$ candidates from
  $C$. We form the preference orders of the voters so that each $c_i$,
  $1 \leq i \leq k-1$, is ranked first by exactly one voter, $c_k$ is
  ranked first by two voters, the voters that do not rank $c_1$ first
  rank $c_1$ last, and each voter ranks $x$ second. Aside from that,
  the preference orders are arbitrary. If $\calR$ had the consensus
  committee property, then $\{c_1, \ldots, c_k\}$ would be the unique
  winning committee.
  Since $\calR$ by Theorem~\ref{separable-kbest} is a best-$k$ rule for scoring vector ${\bf s}=(\row sm)$, with $m=\|C\|$, where 
  $s_i=\gamma(i)$, it suffices to show that $c_1$ gets lower score than $x$
  with respect to ${\bf s}$. Indeed, the score of $x$ is
  $(k+1)s_2$ and the score of $c_1$ is $s_1$. Since we
  assume that $0 < \gamma(1) \leq k\gamma(2)$, we have $s_1\le ks_2$. Moreover, since each
  $c_2, \ldots, c_k$ has ${\bf s}$-score equal or higher than $c_1$, we
  have that $\{c_1, \ldots, c_k\}$ is not a winning committee. Thus
  $\calR$ does not have the consensus committee property.

  Suppose now that $\gamma(1) > k\gamma(2)$. Consider an arbitrary
  election $E = (C,V)$, where there is a group $W$
  of $k$ candidates such that each voter ranks some member of $W$
  first and each member of $W$ is ranked first by either $\lfloor
  \frac{n}{k} \rfloor$ or $\lceil \frac{n}{k} \rceil$ voters. 
  Consider a candidate $x\notin W$. Its ${\bf s}$-score (where
  $\mathbf{s}$ is defined as in the previous paragraph) is at most
  $ns_2$. On the other hand, the score of each candidate in $W$ is at
  least $\lfloor \frac{n}{k}\rfloor s_1$. By assumption,
  $\lfloor\frac{n}{k}\rfloor s_1 > ns_2$ for sufficiently large $n$
  and, thus, $\calR$ satisfies the consensus committee property for
  such $n$.
\end{proof}

\begin{proposition}
  Let $\calR$ be a representation-focused $k$-committee scoring rule
  with committee scoring function  $f(i_1, \ldots, i_k) = \beta(i_1)$.
  Then $\calR$ has the consensus committee property 
  if and only if $\gamma(1) > \gamma(2)$.
\end{proposition}
\begin{proof}
  Consider an election $E$ that satisfies the conditions of the
  consensus committee property.  If $\gamma(1) > \gamma(2)$ then
  clearly the $k$ candidates from $E$ that are ranked first form a
  unique winning committee. On the other hand, if $\gamma(1) =
  \gamma(2)$ then it is possible to form an election with more than
  one winning committee. It suffices to consider a profile of $k$
  votes over $2k$ candidates, where each vote ranks a distinct pair of
  candidates on top. In this case, there are at least $2^k$ different
  winning committees.
\end{proof}

Our final instantiation of Dummett's proportionality for solid
coalitions is the unanimity property. Every committee scoring rule
satisfies its weak variant.

\begin{theorem}
  Every committee scoring rule $\calR$ satisfies weak unanimity.
\end{theorem}
\begin{proof}
  Consider an election $E = (C,V)$ where every voters ranks
  candidates from some set $W$, $\|W\|=k$, on top. Let $f$ be the
  committee scoring function for $\calR$. %
By definition, for each voter $v$ in
  $V$ and each size-$k$ set $Q$ of candidates, we have $f(\pos_v(W)) \geq
  f(\pos_v(Q))$.  Thus, we have $W\in\calR(E, k)$.
\end{proof}

It is immediate that $\ell_1$-Monroe, $\ell_{\min}$-Monroe,
Greedy-Monroe, Bloc and $k$-Borda satisfy strong unanimity, and that
SNTV, $\ell_1$-CC, $\ell_{\min}$-CC, and Greedy-CC do not (to see
this, consider an election where all the voters have the same
preference order; by definition, in those cases these rules choose all
the committees that include the candidate ranked first by all the
voters).

Finally, we note that STV satisfies strong unanimity. If there is some
set $W$ of $k$ candidates that each of the $n$ voters ranks on top,
then in every round of STV there is a candidate from $W$ that is
ranked first by at least $\lfloor \frac{n}{k+1}\rfloor+1$ voters.

\section{Monotonicity}\label{sec:monotonicity}
\noindent
Being monotonic is a natural and easily satisfiable condition for
single-winner rules. Among the few examples of prominent non-monotonic
single-winner rules are STV and Dodgson's
rule~\cite{bra:j:dodgson-remarks}. In contrast, for multiwinner rules
monotonicity is a rather demanding property.  
However, all committee scoring rules satisfy candidate monotonicity,
and all  separable committee scoring rules satisfy non-crossing
monotonicity.

\begin{theorem}\label{thm:csr:monotone}
  Let $\calR$ be a $k$-committee scoring rule. Then $\calR$ satisfies
  candidate monotonicity.
\end{theorem}
\begin{proof}
  Consider an election $E = (C,V)$ with $\|C\|=m$.
 Let $f$ be the $k$-committee scoring
  function defining $\calR$ for $m$ candidates and $k$ winners.
  Let $W$ be a committee in $\calR(E,k)$ and let $c$ be a
  candidate in $W$. Consider a vote $v \in V$ that does not rank $c$
  first, and replace it with a vote $v'$ obtained from $v$ by shifting $c$
  one position forward. Denote the resulting election by $E'$.

  By construction, we have $f(\pos_{v'}(W)) \geq f(\pos_v(W))$. On
  the other hand, for each committee $S\subseteq C\setminus\{c\}$, we
  have $f(\pos_{v'}(S)) \leq f(\pos_v(S))$. Since $W$ was a winning
  committee for $E$, this means that either $W$ is also a winning 
  committee for $E'$ or some committee $W' \in \calR(E',k)$ 
  has a higher score. We must have $c\in W'$, since only committees
  with $c$ can have a higher score in $E'$, compared to $E$.
\end{proof}

\begin{theorem}
  Let $\calR$ be a weakly separable $k$-committee scoring rule. Then
  $\calR$ satisfies non-crossing monotonicity.
\end{theorem}
\begin{proof}
  Let $E = (C,V)$ be an election with $m$
  candidates. %
  Let $f$ be a weakly separable committee scoring function defining
  $\calR$. %
  Let $W$ be a winning committee and let $c\in W$. Consider a vote $v$
  where $c$ is not directly preceded by a member of $W$. Let $d$ be
  the candidate who is directly above $c$ in $v$ ($d \notin W$) and
  let $v'$ be a vote obtained from $v$ by swapping $c$ and $d$.

  After the swap, committees that include $c$ and not $d$ gain the
  same number
  $f(\pos_{v'}(W)) - f(\pos_v(W))$ of points and this number is
  non-negative; every committee with both $c$ and $d$ (or with neither
  $c$ nor $d$) maintains the same score, and every committee with $d$
  but not $c$ loses $f(\pos_{v'}(W)) - f(\pos_v(W)) \geq 0$
  points. Thus $W \in \calR(E',k)$.
\end{proof}

To complete the discussion of committee scoring rules, we need to
consider non-separable committee scoring rules with respect to the
non-crossing monotonicity.  It appears that these rules (such as
$\ell_1$-CC) do not normally satisfy the non-crossing monotonicity.

\begin{proposition}
\label{prop:cc-non-crossing}
  $\ell_1$-CC, $\ell_1$-Monroe, Greedy-CC, and
  Greedy-Monroe fail non-crossing monotonicity.
\end{proposition}
\begin{proof}
Consider election $E = (C,V)$ where $C =
  \{a,b,c,d,x_1, \ldots, x_6\}$ and $V = (v_1, \ldots, v_6)$. The
  preference orders of the voters are as follows:
  \newcommand{\fouurvote}[4]{#1 \succ #2 \succ #3 \succ #4 \succ \cdots}
  \begin{align*}
     v_1\colon& \fouurvote a{x_1}cb, 
     &v_2\colon& \fouurvote a{x_2}db,\\
     v_3\colon& \fouurvote b{x_3}ac,
     &v_4\colon& \fouurvote b{x_4}dc,\\
     v_5\colon& \fouurvote c{x_5}ab, 
     &v_6\colon& \fouurvote c{x_6}db.
  \end{align*}  
  The reader can quickly verify that if we seek a committee of size
  $k=2$ then, according to $\ell_1$-CC, there are three winning
  committees, $\{a,b\}$, $\{a,c\}$, and $\{b,c\}$, each with
  satisfaction $6\|C\|-11$. If we focus on committee $\{a,c\}$,
  non-crossing monotonicity requires this committee to still be
  winning after shifting $c$ forward by one position in
  $v_1$. However, if that happens, the satisfaction of $\{a,c\}$ does
  not change but the satisfaction of $\{b,c\}$ increases to
  $6\|C\|-10$.  The same construction works for $\ell_1$-Monroe, Greedy-CC, and
  Greedy-Monroe. %
\end{proof}
While $\ell_{\min}$-CC is not a committee scoring rule, it behaves
similarly to non-separable committee scoring rules. However, this is
not the case for $\ell_{\min}$-Monroe. %

\begin{theorem}\label{thm:ell_min-mon}
  $\ell_{\min}$-CC satisfies candidate monotonicity, but $\ell_{\min}$-Monroe fails it.
  Both $\ell_{\min}$-CC and $\ell_{\min}$-Monroe fail non-crossing monotonicity.
\end{theorem}
\begin{proof}
  The proof of the first part is similar to that of
  Theorem~\ref{thm:csr:monotone}. If $W$ is a winning committee for
  $\ell_{\min}$-CC and $c$ is a candidate in $W$, then shifting $c$
  forward has the following effect. The aggregate satisfaction of
  every committee that includes $c$ either increases by one, or stays
  the same. The aggregate satisfaction of every committee not
  containing $c$ either stays the same or decreases by one. Thus $c$
  belongs to at least one winning committee. To show that
  $\ell_{\min}$-Monroe fails candidate monotonicity, consider an
  election with candidate set $C = \{a,b,c,d\}$ and four
  voters %
  $V = (v_1,v_2,v_3, v_4, v_5, v_6)$, with preference orders
  as follows:
  \begin{align}
  \label{v4}
     v_1\colon& \fourvote acdb,  
   & v_2\colon& \fourvote cadb,\nonumber\\
     v_3\colon& \fourvote bdca, 
   & v_4\colon& \fourvote dbac,\\
     v_5\colon& \fourvote dacb, 
   & v_6\colon& \fourvote bcda\nonumber.
  \end{align}
  For $k=2$, the winning committees are $\{a,b\}$ and $\{c,d\}$,
  both with satisfaction $\|C\|-2$. If we shift $a$ forward by one
  position in $v_4$, the satisfaction of $\{a,b\}$ decreases to
  $\|C\|-3$ but the satisfaction of $\{c,d\}$ stays the
  same.

  To see that $\ell_{\min}$-CC fails non-crossing monotonicity,
  consider election $E = (C,V)$ with $C = \{a,b,c,x_1, \ldots,x_{11}\}$
  and $V = (v_1, v_2,v_3, v_4, v_5, v_6)$, where the voters have the following
  preference orders: %
  \newcommand{\threevote}[4]{#1 \succ #2 \succ #3 \succ #4 \succ
    \cdots}
  \begin{align*}
     v_1\colon& \threevote a{x_1}{x_2}b, 
     &v_2\colon& \threevote {x_3}{x_4}{x_5}c,\\
     v_3\colon& \threevote bac{x_6},
     &v_4\colon& \threevote bac{x_7},\\
     v_5\colon& \threevote bc{x_8}{x_9},
     &v_6\colon& \threevote ac{x_{10}}{x_{11}}.
  \end{align*}
  We seek a committee of size $k=2$.  There are four tied ones,
  $\{b,c\}$ and $\{a,c\}$, $\{x_1,c\}$ and $\{x_2,c\}$, all with
    aggregate satisfaction $\|C\|-4$. These are the
    winners. Non-crossing monotonicity requires that if we shift $c$
    forward in $v_2$, then $\{b,c\}$ should still be a winning
    committee. However, this committee's satisfaction stays the same,
    whereas the satisfaction of $\{a,c\}$ increases to $\|C\|-3$.
  The same construction applies to $\ell_{\min}$-Monroe. %
\end{proof}

The remaining multiwinner rules studied in this paper fail each of our
monotonicity criteria. For STV this is well-known to happen even for $k=1$.
For the rest of the rules, we  provide the following result.

\begin{proposition}
  $\ell_1$-Monroe, Greedy-Monroe, and Greedy-CC fail candidate monotonicity.
\end{proposition}

\begin{proof}
  First, we note that the argument given in
  Theorem~\ref{thm:ell_min-mon} to show that $\ell_{\min}$-Monroe is not
  candidate monotone (profile~\eqref{v4} and moving $a$ forward in
  vote $v_4$) in fact also applies to $\ell_1$-Monroe.

  We move on to Greedy-Monroe. Let $E = (C,V)$ be an election with $C
  = \{a,b,c,d\}$ and $V = (v_1, \ldots, v_8)$ whose preference orders
  are:
  \begin{align*}
     v_1\colon&  \fivevote bcdae,
     &v_2\colon& \fivevote dcbae,\\
     v_3\colon&  \fivevote aedbc,
     &v_4\colon& \fivevote abdce,\\
     v_5\colon&  \fivevote aedbc,
     &v_6\colon& \fivevote bdace,\\
     v_7\colon&  \fivevote dcbae,
     &v_8\colon& \fivevote cbdae.
  \end{align*}
  For $k = 2$, under Greedy-Monroe there are two winning committees
  $\{a,c\}$ and $\{b,d\}$. This is so because in the first iteration,
  Greedy-Monroe picks either $a$ or $b$. If it picks $a$, then in the
  next iteration it picks $c$. If it picks $b$, then irrespective
  which voters it chooses to assign to $b$ in the first iteration
  (among the assignments allowed under Greedy-Monroe), in the second
  iteration it picks $d$. However, if we shift $c$ forward by one
  position in $v_6$ then only $\{b,d\}$ remains winning (Greedy-Monroe
  no longer can pick $a$ in the first iteration). This shows that
  Greedy-Monroe is not candidate monotone.

  For the case of Greedy-CC, consider the election $E = (C,V)$ with
  $C = \{a,b,c,d\}$ and $V = (v_1, \ldots, v_6)$, where 
  \begin{align*}
     v_1\colon&  \fourvote abcd,
     &v_2\colon& \fourvote bcda,\\
     v_3\colon&  \fourvote abcd,
     &v_4\colon& \fourvote cbda,\\
     v_5\colon&  \fourvote abcd,
     &v_6\colon& \fourvote dacb.
  \end{align*}
  For $k=2$, Greedy-CC %
  declares $\{a,b\}$ and $\{a,c\}$ as winners. (In the first iteration
  Greedy-CC picks either $a$ or $b$. In the former case, in the second
  iteration it picks either $b$ or $c$. In the latter case, i.e., if
  it picks $b$ in the first iteration, in the second iteration it
  picks $c$). However, if we shift $c$ forward by one position in
  $v_6$, then $\{a,b\}$ and $\{a,d\}$ are winning (in this case, in
  the first iteration Greedy-CC has to pick $b$, and then in the
  second iteration there is a tie between $a$ and $d$). Thus Greedy-CC
  fails candidate monotonicity as well.
\end{proof}

\section{Consistency and Homogeneity}\label{sec:con-hom}
\noindent
For single-winner rules, the famous Young's theorem~\cite{you:j:scoring-functions} says that only
scoring rules and their compositions satisfy
consistency. While we do not know
how to extend this result to multiwinner rules, the
situation seems to be similar: We show that every committee
scoring rule satisfies consistency, whereas other rules fail it.

\begin{proposition}
  Every committee scoring rule satisfies consistency. In particular,  $\ell_1$-CC is consistent.
\end{proposition}
\begin{proof}
  Let $\calR$ be a $k$-committee selection scoring rule %
and let $f$ be the corresponding
  committee scoring function. Let $E_1$ and $E_2$ be two elections,
  $E_1 = (C,V_1)$ and $E_2 = (C,V_2)$, where $\|C\| = m$, and where
  $\calR(E_1,k) \cap \calR(E_2,k) \neq \emptyset$.

  Let $W$ %
  be a committee such that $W \in
  \calR(E_1,k)$ and $W \in \calR(E_2,k)$.
  We claim that $W \in \calR(E_1+E_2)$. If it were not the case, then
  there would have to be a committee $Q$ %
  of size $k$ such that either the score of $Q$ in $E_1$ or the score of $Q$
  in $E_2$ were higher than that of $W$ in these elections. This
  contradicts the choice of $W$.

  Similarly, if $T$ is a winning committee for $E_1+E_2$ then it
  must be the case that $T \in \calR(E_1,k)$ and $T \in
  \calR(E_2,k)$. This is so, because the score of $T$ in $E_1+E_2$ is
  the same as that of $W$. However, if $T$ were not in $\calR(E_1,k)$,
  then this would mean that the score of $T$ in $E_1$ is smaller than
  the score of $W$ in $E_1$ and, so, the score of $T$ in $E_2$ is
  higher than that of $W$ in $E_2$.  This would mean that it is not
  the case that $W \in \calR(E_1,k)$ and $W \in \calR(E_2,k)$.
\end{proof}

Now let us turn our attention to the negative results. We note that
some of our rules, such as STV, $\ell_{\min}$-CC and
$\ell_{\min}$-Monroe, can be ruled out to be consistent on the basis
of Young's result because for $k=1$ these rules are not scoring rules.

\begin{corollary}
  Neither of STV, $\ell_{\min}$-CC and $\ell_{\min}$-Monroe is
  consistent.
\end{corollary}

On the other hand, for $k=1$ each of $\ell_1$-Monroe, Greedy-CC and
Greedy-Monroe is equivalent to Borda, which is consistent. For these
rules, we directly show that they are not consistent.

\begin{proposition}
  Neither of 
  $\ell_1$-Monroe, Greedy-CC and Greedy-Monroe is consistent.
\end{proposition}
\begin{proof}
  Consider $\ell_1$-Monroe. Let $E_1 = (C,V_1)$ and $E_2 = (C,V_2)$ be
  two elections, where $C = \{a,b,c,d,e\}$, $V_1 = (v_1, v_2,v_3, v_4)$
  and $V_2 = (v_5, v_6,v_7, v_8)$.  The voters in $V_1$ have the
  following preferences: 
  \newcommand{\vote}[5]{#1 \succ #2 \succ #3 \succ #4 \succ #5}
  \begin{align*}
     v_1\colon&  \vote acdeb,
     &v_2\colon& \vote acdeb,\\
     v_3\colon&  \vote acdeb,
     &v_4\colon& \vote becda.
  \end{align*}
  The voters in $V_2$ have the following preferences:
  \begin{align*}
     v_5\colon&  \vote abcde,
     &v_6\colon& \vote abcde,\\
     v_7\colon&  \vote cdbea,
     &v_8\colon& \vote bcdea.
  \end{align*}
  We seek a committee of size $k=2$.  It is straightforward to check that under $\ell_1$-Monroe, committee
  $\{a,c\}$ would win in both $E_1$ and $E_2$ but in $E_1+E_2$,
  committee $\{a,b\}$ has higher satisfaction than $\{a,c\}$. %

  For Greedy-Monroe, we use $C = \{a,b,c,d\}$ and elections $E_1 =
  (C,V_1)$ and $E_2 = (C,V_2)$, where $V_1$ and $V_2$ have two voters
  each. The voters in $V_1$ have the following preferences:
  \newcommand{\votef}[4]{#1 \succ #2 \succ #3 \succ #4}
  \begin{align*}
     v_1\colon&  \votef abcd,
     &v_2\colon& \votef abcd.
  \end{align*}
  The voters in $V_2$ have the following preference orders:
  \begin{align*}
     v_3\colon&  \votef acdb,
     &v_4\colon& \votef bcda.
  \end{align*}
  For $k=2$, Greedy-Monroe chooses $\{a,b\}$ as the winning committee
  in both elections, but for $E_1+E_2$ it chooses $\{a,b\}$ and
  $\{a,c\}$ (due to the parallel-universe tie-breaking). This shows
  that Greedy-Monroe is not consistent.
  
  For Greedy-CC we, again, take $C = \{a,b,c,d\}$ and elections $E_1 = (C,V_1)$
  and $E_2 = (C,V_2)$ but this time $V_1$ and $V_2$ contain four voters each. 
  The voters in $V_1$ have the following preference orders:
  \begin{align*}
     v_1\colon&  \votef acdb,
     &v_2\colon& \votef acdb, \\
     v_3\colon&  \votef acdb,
     &v_4\colon& \votef bcda.
  \end{align*}
  The voters in $V_2$ have the following preference orders:
  \begin{align*}
     v_5\colon&  \votef bcda,
     &v_6\colon& \votef bcda, \\
     v_7\colon&  \votef bcda,
     &v_8\colon& \votef acdb.
  \end{align*}
  For $k=2$, in both $E_1$ and $E_2$, Greedy-CC picks $\{a,b\}$.
  However, in $E_1 + E_2$ it picks $\{a,c\}$ and $\{b,c\}$ (to see
  that Greedy-CC is non consistent it suffices to note that in
  $E_1+E_2$ it has to pick $c$ in the first iteration because $c$ is
  the unique Borda winner in this election).  %
\end{proof}

We now consider homogeneity.  Naturally, committee
scoring rules are homogeneous because consistency implies 
homogeneity. For other rules, 
the situation is more complex.

\begin{theorem}
  Both $\ell_{\min}$-CC and Greedy-CC satisfy homogeneity.
\end{theorem}
 We omit these proofs since they are straightforward.
 Interestingly, neither of the variants of Monroe's rule is
 homogeneous, although, as we will see in
 Theorem~\ref{thm:homogeneity}, $\ell_1$-Monroe, $\ell_{\min}$-Monroe
 are partially homogeneous.
\begin{proposition}
  $\ell_1$-Monroe, $\ell_{\min}$-Monroe, and Greedy-Monroe are not
  homogeneous.
\end{proposition}
\begin{proof}
  Consider an election with candidate set $C = \{a,b,c,d\}$ and three
  voters %
 $v_1$, $v_2$, and $v_3$, with the following preference orders:
  \newcommand{\vote}[4]{#1 \succ #2 \succ #3 \succ #4}
  \begin{align*}
     v_1\colon&  \vote abdc,
     &v_2\colon& \vote abdc,
     &v_3\colon& \vote cbda.
  \end{align*}
  We seek a committee of size $k = 2$.  
For $\ell_1$-Monroe,
  Greedy-Monroe and $\ell_{\min}$-Monroe, the unique
  winning committee of size $2$ is $\{a,c\}$. However, for $2E$ the winner sets for these rules
  include $\{a,b\}$.
\end{proof}

On the positive side, if the number of voters is divisible by the
size of the committee, then $\ell_1$-Monroe and $\ell_{\min}$-Monroe
are homogeneous. In essence, this means that variants of Monroe fail
homogeneity due to rounding problems in the Monroe criterion. One
solution would be to clone each
voter $k$ times when seeking a committee of size $k$. 
We do not consider this modification of Monroe's rule
here, but it would be interesting to see how the satisfaction of a
committee elected in this way compares to that elected without 
cloning. %

\begin{theorem}
\label{thm:homogeneity}
  Both $\ell_1$-Monroe and $\ell_{\min}$-Monroe satisfy homogeneity,
  provided that the number of voters $n$ in the election is divisible by
  the size $k$ of the committee to be selected.
\end{theorem}
\begin{proof}
  Let $\calR\in\{\ell_1$-Monroe, $\ell_{\min}$-Monroe$\}$, pick
  an election $E = (C,V)$, and let $k$ be a positive integer 
  that divides $n=\|V\|$. We will show that 
  $\calR(E,k) = \calR(tE,k)$ for each $t>0$.

  Let $W$ be a
  committee that wins in $tE$. %
  We refer to the members of $W$ as the
  \emph{winners}. Let $\Phi \colon tV \rightarrow C$ be an assignment
  of candidates to voters witnessing that %
   $W\in\calR(tE,k)$. By Monroe's criterion, for each 
  $w \in W$ we have $\|\Phi^{-1}(w)\| = \frac{nt}{k}$.  
  We now proceed as follows. First, we show how to
  transform $\Phi$ into an assignment $\Phi'$ such that (a) under
  $\Phi'$ each winner represents exactly $\frac{n}{k}$ voters in each
  copy of $V$ and (b) the satisfaction of the voters under $\Phi'$
  is the same as under $\Phi$. We then prove
  that $\Phi'$ can be further transformed into $\Phi''$ that uses the
  same assignment for each copy of $V$. 

  Let $V_1, \ldots, V_t$ be $t$ copies of $V$ so
  that $tV = V_1 + \cdots + V_t$; we assume that within each $V_i$,
  $i\in[t]$, voters are listed in the same order.  For each $i\in[t]$,
  $\ell\in[n]$, we write $v_{i,\ell}$ to denote the $\ell$-th voter in
  $V_i$.  Our proof relies on the observation that for each $\ell \in
  [n]$ and each $i,j \in [t]$, if we take some assignment function
  $\Phi$ and swap its values for $v_{i,\ell}$ and $v_{j,\ell}$, then
  we get an assignment function $\Phi'$ with the same societal
  satisfaction.  To exploit this fact, for each $\ell \in [n]$, we
  write $\rep(\ell) = \{ \Phi(v_{i,\ell}) \mid i \in [t] \}$ to denote
  the set of the representatives (under $\Phi$) of the set of $\ell$'th
  voters in the profiles $V_1, V_2, \ldots, V_t$.

  We now show that, using the transformations just described, it is possible to transform $\Phi$ into an
  assignment $\Phi^{(1)}$ which assigns each winner from $W$ to
  exactly $\frac{n}{k}$ voters from $V_1$. We use the fact that for
  each $\ell \in [n]$, we can assign an arbitrary member of
  $\rep(\ell)$ to voter $v_{1,\ell}$ by swapping the representative of
  $v_{1,\ell}$ with the representative of $v_{j,\ell}$, for
  appropriately selected $j$. We build the following bipartite graph
  (the vertices are partitioned into the set of voter vertices and the
  set of winner vertices):
  \begin{enumerate}
  \item The voter vertices are exactly the voters from $V_1$.
  \item For each winner $w \in W$, we have $\frac{n}{k}$ winner
    vertices $w^1, \ldots, w^{\frac{n}{k}}$ which we will call {\em clones} of $w$.
  \item For each voter vertex $v_{1,\ell} \in V_1$ and each winner $w
    \in W$, there are edges between $v_{1,\ell}$ and each of the
    winner vertices $w^1, \ldots, w^{\frac{n}{k}}$, if $w \in
    \rep(\ell)$, and no such edges exist, if $w \notin
    \rep(\ell)$.
  \end{enumerate}
  It is clear that if there is a perfect matching between the voter
  vertices and the winner vertices in this graph, then it is possible
  to transform $\Phi$ into $\Phi^{(1)}$ that assigns each winner to
  exactly $\frac{n}{k}$ voters in $V_1$. Given such a perfect
  matching, we transform $\Phi$ so that for each voter $v_{1,\ell}$
  that is matched to some vertex $w^u$, we swap this voter's
  representative (if it is not  already $w$) with the representative  of some voter $v_{j,\ell}$
  for whom $\Phi(v_{j,\ell}) = w$. This is possible since $w\in \rep(\ell)$. Thus it remains to show that our
  graph has a perfect matching. To this end, we use the famous Hall's
  theorem.

  For each subset $V'$ of voter vertices, we define the set of
  neighbors of $V'$, denoted $N(V')$, to be the set of those winner
  vertices that are connected to some member of $V'$. Hall's theorem
  says that there is a perfect matching in our graph if and only if
  for every set $V'$ of voter vertices we have $\|N(V')\| \geq
  \|V'\|$.

  Let $V'$ be some arbitrary subset of voter vertices.  By the
  construction of our graph, we know that if $N(V')$ contains one clone of $w$, then it contains all of them, hence $\|N(V')\|$ is of the form
  $q\frac{n}{k}$, where $q$ is a positive integer, and that $N(V')$
  corresponds to a set of $q$ winners, each of whom has $\frac{n}{k}$
  clones in $N(V')$. This means that in $tV$ there is a group of
  $t\|V'\|$ voters represented under $\Phi$ by $q$ winners.  Since in
  $tV$ each winner is assigned to exactly $t\frac{n}{k}$ voters, it
  must be that $t\|V'\| \leq qt\frac{n}{k}$. This is
  equivalent to $\|V'\| \leq q\frac{n}{k} = \|N(V')\|$, the
  requirement of the Hall's theorem. Thus there is a perfect matching
  for our graph and the desired assignment $\Phi^{(1)}$ exists.

  Now, removing $V_1$ from consideration and repeating the above
procedure for $V_{2}$, then $V_{3}$, and so on, we can eventually transform
  $\Phi^{(1)}$ into $\Phi'$, where each winner represents exactly
  $\frac{n}{k}$ voters from each sequence of voters $V_{i}$, $i \in
  [t]$. It is then easy to see that we can replace $\Phi'$ with
  $\Phi''$ that uses the same assignment of winners to voters within
  each copy of $V_i$ (each $V_i$ has an assignment with the same
  societal satisfaction; otherwise the original assignment $\Phi$
  would not have been optimal). This implies that $W \in \calR(E,k)$
  because otherwise there would be some set $W' \in \calR(E,k)$ that
  would give lower satisfaction to the voters in $tE$ than $W$
  gives. Analogously, this means that every $W' \in \calR(E,k)$ also
  belongs to $\calR(tE,k)$.
\end{proof}

\begin{proposition}
  Greedy-Monroe fails homogeneity even if the size of the committee divides the number of voters.
\end{proposition}
\begin{proof}
  Let $E = (C,V)$ be an election with $C = \{a,b,c,d\}$ and $V = \{v_1, \ldots, v_6\}$.
  The voters have the following preference orders:
  \newcommand{\vote}[4]{#1 \succ #2 \succ #3 \succ #4}
  \begin{align*}
     v_1\colon&  \vote abcd,
     &v_2\colon& \vote badc,
     &v_3\colon&  \vote abcd,\\
     v_4\colon& \vote bcda,
     &v_5\colon&  \vote cadb,
     &v_6\colon& \vote dacb.
  \end{align*}
  We seek a committee of size $k=2$.  Under Greedy-Monroe, the
  tied-for-winning committees are $\{a,b\}$ and $\{a,c\}$. 
  Indeed, in the first iteration Greedy-Monroe may pick $a$. It
  can assign $a$ to %
  $(v_1,v_2,v_3)$, $(v_1,v_3,v_5)$, or $(v_1,v_3,v_6)$.  Depending on
  the choice, in the second iteration it picks either $c$ or $b$. In the first iteration it can also pick $b$, then in the second iteration it picks $a$.

  Now consider election $2E$. For each $i\in\{1,\dots, 6\}$, $j\in\{1,
  2\}$ we write $v_i^j$ to refer to $v_i$ in the $j$-th copy of~$E$.
  In the first iteration, Greedy-Monroe again is allowed to pick $a$
  (due to parallel-universes tie-breaking). By parallel-universes
  tie-breaking, it is allowed to assign $a$ to
  $v_1^1,v_1^2,v_2^1,v_3^1,v_3^2,v_5^1$. Thus, in the second
  iteration, the remaining votes are:
  \begin{align*}
     v^1_6\colon& \vote dacb,
     &v^2_2\colon& \vote badc,
     &v^2_5\colon& \vote cadb,\\
     v^1_4\colon& \vote bcda,
     &v^2_4\colon& \vote bcda,
     &v^2_6\colon& \vote dacb.
  \end{align*}
  The unique Borda winner of this election is $d$, so
  Greedy-Monroe picks $d$ in the second iteration. This means that
  $\{a,d\}$ is a winning committee in $2E$, a contradiction.
\end{proof}

The above proposition relies heavily on parallel-universes tie-breaking.
It is possible to refine the intermediate tie-breaking
procedure of Greedy-Monroe so that it becomes homogeneous when $k$
divides $\|V\|$. We omit the details here.

\section{Related Literature}\label{sec:literature}

With all the preceding discussions in hand, it is high time to discuss
how our results and approaches compare to others in the literature.
Unfortunately, the literature on the properties of multiwinner rules
is still relatively sparse, compared to that regarding single-winner
rules, and is scattered between different fields of research, ranging
from behavioral science, through political science, social choice
theory, to computer science.  Here we review those papers that are in
spirit closest to our work.

The paper most closely related to ours is that of Felsenthal and Maoz
\cite{fel-mao:j:norms}. They consider four $k$-choice functions, the
Plurality rule (i.e., in our terminology, the SNTV rule), the Approval
rule,\footnote{This rule is not preference-based} the Borda rule
(i.e., $k$-Borda), and STV. They adapt a range of single-winner
normative properties to the multiwinner setting and study them in the
context of these rules. The main two differences between our paper and
theirs are as follows. First, we study a somewhat different set of
rules (in particular we do not consider the Approval rule, but we do
consider Bloc and a number of rules focused on proportional
representation).  Second, the set of axioms that we consider is closer
in spirit to the goal of achieving proportional representation than
theirs. On the one hand, we introduce some new axioms that try to
capture proportional representation (such as consensus committee and
solid coalitions properties), and, on the other hand, we do not
consider axioms related to the Condorcet principle (and Felsenthal and
Maoz do).  Naturally, our papers also have many similarities.  Both we
and Felsenthal and Maoz consider monotonicity, though our view is
slightly more detailed (we study two variants, candidate monotonicity
and non-crossing monotonicity).
Both us and Felsenthal and Maoz~\cite{fel-mao:j:norms} study committee
monotonicity, which they call continuity. We prefer our name since
continuity in social choice refers to a different
property~\cite{you:j:scoring-functions}.  Finally, both Felsenthal and
Maoz and us study the consistency property, which is crucial when one
considers rules based on scoring protocols (indeed, as shown by
Young~\cite{you:j:scoring-functions}, up to some tie-breaking
specifics, scoring protocols are the only single-winner rules that are
consistent, and, thus, it is natural to expect that a similar result
holds in the multiwinner setting; very interestingly, consistency also
plays an important role in Bock et al.'s study of consensus-based
multiwinner rules~\cite{boc-day-mcm:j:consensus-committee}).

Young's consistency-based characterization of scoring rules, as well
as his characterization of the Borda rule, have already inspired
researchers working on multiwinner rules.  For example, Debord
extended Young's characterization of Borda to the case of $k$-Borda,
by showing that it is the rule that satisfies the axioms of
neutrality, faithfulness, consistency, and the cancellation property
(we point readers to his work for details on these
properties~\cite{deb:j:k-borda}). In his follow-up work, Debord also
introduced some new axioms for multiwinner rules.  These axioms,
however, are geared toward the case of rules that elect what he calls
$k$-elites~\cite{deb:j:prudent}, which are multiwinner analogs of
Condorcet winners.  
Indeed, there is a number of multiwinner rules that focus on electing
a committee by following the principle of Condorcet consistency;
related issues were considered, for example, by Barber\`a and
Coelho~\cite{bar-coe:j:non-controversial-k-names}, Kaymak and
Sanver~\cite{kay-san:j:condorcet-winners}
Fishburn~\cite{fis:j:majority-committees},
Ratliff~\cite{rat:j:condorcet-inconsistencies}, and others. Since we
do not consider Condorcet rules, we do not present this literature in
depth.

The study of committee selection rules in approval voting framework is
well-advanced. Kilgour~\cite{kil-handbook} describes a number of
approval-based voting rules that elect a committee of fixed size and
proved some of their basic properties, Kilgour and Marshall
\cite{kil-mar:j:minimax-approval} give a nice survey of approval-based
committee selection rules and add some new ones. Aziz et al
\cite{aziz-elk-etc} have recently introduced an axiom called Justified
Representation and showed that in their framework the only committee
selection rule that satisfies it is Proportional Approval Voting
(PAV). Again, we do not focus on approval-based rules so we point the
readers to the above-cited papers for more detailed discussions.

Multiwinner voting often raises controversies. For example, Staring
\cite{sta:j:paradoxes} demonstrates that Bloc rule is extremely
committee non-monotonic (he calls it {\em increasing-committee-size
  paradox}). He demonstrates a profile where the winning committee of
size three is disjoint from the winning committee of size two, and the
winning committee of size four is disjoint from those for sizes two
and three. Similar phenomena is observed by
Ratliff~\cite{rat:j:condorcet-inconsistencies} for other rules, who
even titled his paper ``Some Startling Inconsistencies when Electing
Committees'' (we mention that, in another paper, he also points to some
possibly undesirable behavior in other multiwinner
settings~\cite{rat:j:selecting-committees}). On the other hand, in
this paper we view the fact that a rule fails to be committee
monotonic as a rather expected feature of every rule that attempts to
achieve proportional representation of the voters. This is yet another
reason to believe that multiwinner rules can only be judged in the
context of their applications. A rule for selecting a parliament
should satisfy quite different desiderata than a rule used for
shortlisting.

Multiwinner voting rules are more complex than the single-winner ones
also from the computational standpoint. While there are some
computationally hard single-winner rules (most notably the rules of
Kemeny, Young, and
Dodgson~\cite{bar-tov-tri:j:who-won,hem-hem-rot:j:dodgson,hem-spa-vog:j:kemeny,rot-spa-vog:j:young,bet-guo-nie:j:dodgson-parametrized}),
for the case of multiwinner voting, this problem seems to be much
more pronounced.  For example, Darmann~\cite{dar:j:condorcet-hard}
showed that testing if a given committee satisfies a particular
Condorcet criterion is $\np$-hard. Similarly, it is $\np$-hard to
compute winners under Monroe's
rule~\cite{pro-ros-zoh:j:proportional-representation,bet-sli-uhl:j:mon-cc},
Chamberlin--Courant's
rule~\cite{pro-ros-zoh:j:proportional-representation,bou-lu:c:chamberlin-courant,bet-sli-uhl:j:mon-cc},
and many other similar rules~\cite{sko-fal-lan:c:multiwinner}.  The
same holds for a number of approval-based rules such as
PAV~\cite{sko-fal-lan:c:multiwinner,azi-gas-gud-mac-mat-wal:c:multiwinner-approval}
or Minimax Approval rule~\cite{leg-mar-meh:c:minimax-approval} of
Brams et al.~\cite{bra-kil-san:j:minimax-approval}, or even for STV
with parallel-universes tie-breaking~\cite{con-rog-xia:c:mle}. In
effect, many researchers have sought efficient approximate algorithms
for these rules, many of which can be viewed as full-fledged voting
rules in their own right. Indeed, this is how the
Greedy-CC~\cite{bou-lu:c:chamberlin-courant} and
Greedy-Monroe~\cite{sko-fal-sli:j:multiwinner} rules that we study
came to be.  Nonetheless, there is also a number of natural
multiwinner rules that are polynomial-time computable. These include,
for example, rules based on electing candidates with the $k$ highest
scores according to some scoring protocol (e.g., SNTV, Bloc,
$k$-Borda), but also some other, more complicated rules (see, e.g.,
the work of Klamler et al.~\cite{kla-pfe-ruz:j:weighted-committee}).

\section{Conclusions}\label{sec:conclusions}
\noindent
We formalized a number of natural properties for multiwinner voting
rules and conducted a comprehensive comparison of ten prominent
multiwinner rules with respect to these properties.  The choice of
these rules was guided by the fact that each of them, in some broad
sense, is based on some single-winner scoring rule.  In the course of
our study, we identified two natural families of multiwinner rules,
best-$k$ rules and committee scoring rules. The latter class is
particularly interesting because it contains both rules of low
complexity, such as $k$-Borda, SNTV, and Bloc, and rules with hard
winner-determination problems, such as $\ell_1$-CC.
We have put forward a framework for studying multiwinner rules and
considered a number of their properties.  We believe that our results
give a better understanding of applicability of various multiwinner
rules to particular tasks. For example, we see that best-$k$ rules are
well-suited for picking a group of finalists in a competition, whereas
rules based on the Monroe criterion ($\ell_1$-Monroe,
$\ell_{\min}$-Monroe, and Greedy-Monroe), as well as STV, seem to be
more appropriate for applications that require proportional
representation (e.g., parliamentary elections). In this context,
Greedy-Monroe is particularly interesting. It was derived as an
approximation algorithm for
$\ell_1$-Monroe~\cite{sko-fal-sli:j:multiwinner}, but it 
has
more appealing properties than the original
rule. %
We believe that Greedy-Monroe should be taken as %
a full-fledged voting rule.

Our results for $\ell_1$-CC and $\ell_{\min}$-CC are similar to those
for Monroe, but intuitively these rules are better suited for
applications such as movie selection (see the introduction) than, say,
parliamentary elections. The reason is that they may assign very
different numbers of voters to each winning candidate (naturally, one
could imagine rules for parliamentary elections where voters would be
represented by more than a single person---and thus different winning
candidates might represent different numbers of voters---but
$\ell_1$-CC and $\ell_{\min}$-CC do not operate on such basis).
Thus, if $\ell_1$-CC or $\ell_{\min}$-CC were used to elect a
parliament, this parliament would have to use weighted voting in its
operation. A recommendation system, on the other hand, simply needs to
present users with a ``committee'' of items so that as many customers
(voters) as possible would find at least one of them satisfying.
In this light, it is a bit disappointing that Greedy-CC, which was
designed as an approximation algorithm for $\ell_1$-CC, does not seem
to perform very well in our comparison.  Indeed, it fails to satisfy
the solid coalitions property (like $\ell_1$-CC but unlike
Greedy-Monroe) and it fails to satisfy the consensus committee
property (unlike every other rule that focuses on some form of
proportional representation).  This latter fact can be seen as a
consequence of Greedy-CC satisfying committee monotonicity (which we
argued to be undesirable from the point of view of proportional
representation).
After comparing Greedy-CC and Greedy-Monroe, it is tempting to simply
use Greedy-Monroe in Greedy-CC's place. 
However, perhaps a better idea would be to modify Greedy-Monroe to not
respect the Monroe's criterion.  Then the challenge would be to
optimally pick the numbers of voters that Greedy-Monroe considers in
each iteration. We believe that studying such variant of Greedy-Monroe
is an important research direction.
Finally, our results regarding SNTV are quite interesting.  It turns
out that, among all the rules that we consider, SNTV satisfies all the
properties that we defined (though it only satisfies unanimity in the
weak sense). One explanation for this fact is that it can be viewed as
an approximation of the rule that simply picks all the candidates that
are ranked first by at least one voter, and, in this sense, is close
to implementing direct democracy. Yet one should not forget that SNTV
ignores each voter's preferences beyond the top candidate and, thus,
inherits all negative features of the Plurality rule.

\bigskip

\noindent\textbf{Acknowledgements.} 
Piotr Faliszewski was supported in part by NCN grants
2012/06/M/ST1/00358 and 2011/03/B/ST6/01393, and by the Polish
Ministry of Science and Higher Education (under AGH University Grant
11.11.230.015 (statutory project)). Arkadii Slinko was supported by NZ Marsden Fund grant
(grant no. UOA 254).

\bibliographystyle{plain} 
\bibliography{grypiotr2006}

\end{document}